\documentclass[lettersize,journal]{IEEEtran}

\usepackage{cite}
\usepackage{algorithmic}
\usepackage{amsmath,amsfonts}
\usepackage{algorithm}
\usepackage{array}
\usepackage[caption=false,font=normalsize,labelfont=sf,textfont=sf]{subfig}
\usepackage{textcomp}
\usepackage{stfloats}
\usepackage{url}
\usepackage{verbatim}
\usepackage{graphicx}
\usepackage{amsthm}
\usepackage{booktabs}
\usepackage{xcolor}
\newtheorem{myremark}{Remark}
\newtheorem{mydefinition}{Definition}
\newtheorem{mylemma}{Lemma}
\newtheorem{mytheorem}{Theorem}
\begin{document}

\setlength{\abovecaptionskip}{1pt}
\setlength{\belowcaptionskip}{0pt}

\title{Approaches to Analysis and Design of AI-Based Autonomous Vehicles}

\author{Tao Yan,~\IEEEmembership{Member,~IEEE}, Zheyu Zhang, Jingjing Jiang,~\IEEEmembership{Member,~IEEE}, \\ and Wen-Hua Chen,~\IEEEmembership{Fellow,~IEEE}
\thanks{This work was supported by the UK Engineering and Physical
Sciences Research Council (EPSRC) Established Career Fellowship under
the grant number EP/T005734/1. \textit{(Corresponding author: Wen-Hua Chen.})}
\thanks{Tao Yan, Zheyu Zhang, Jingjing Jiang, and Wen-Hua Chen are with the Department of Aeronautical and Automotive Engineering, Loughborough University, Loughborough, LE113TU, U.K.~ \{t.yan, z.zhang8, j.jiang2,
w.chen\}@lboro.ac.uk }}




\maketitle

\begin{abstract}
Artificial intelligence (AI) models are becoming key components in an autonomous vehicle (AV), especially in handling complicated perception tasks. However, closing the loop through AI-based feedback may pose significant risks on reliability of autonomous driving due to very limited understanding about the mechanism of AI-driven perception processes. To overcome it, this paper aims to develop tools for modeling, analysis, and synthesis for a class of AI-based AV; in particular, their closed-loop properties, e.g., stability, robustness, and performance, are rigorously studied in the statistical sense. First, we provide a novel modeling means for the AI-driven perception processes by looking at their error characteristics. Specifically, three fundamental AI-induced perception uncertainties are recognized and modeled by Markov chains, Gaussian processes, and bounded disturbances, respectively. By means of that, the closed-loop stochastic stability (SS) is established in the sense of mean square, and then, an SS control synthesis method is presented within the framework of linear matrix inequalities (LMIs). Besides the SS properties, the robustness and performance of AI-based AVs are discussed in terms of a stochastic guaranteed cost, and criteria are given to test the robustness level of an AV when in the presence of AI-induced uncertainties. Furthermore, the stochastic optimal guaranteed cost control is investigated, and an efficient design procedure is developed innovatively based on LMI techniques and convex optimization. Finally, to illustrate the effectiveness, the developed results are applied to an example of car following control, along with extensive simulation.
\end{abstract}


\section{Introduction}
\IEEEPARstart{A}{UTONOMOUS} vehicles (AVs) have received increasing attention over the past decade in both industry and academia, due to the significant advancements in artificial intelligence (AI), especially the deep neural networks which greatly improve the sensing and perception (S\&P) capabilities and level up the autonomy of machines. While those deep AI models exhibit extremely better efficiency and accuracy than traditional S\&P approaches in tasks like classification and regression, new critical challenges will bring in when AI-based S\&P processes are integrated into an autonomous vehicle\cite{Margarita,Marti,Nascimento}. For instance, most AI models are trained in a black-box manner and  validated on limited datasets, and this makes it hard for engineers to grasp their robustness when something unseen or erroneous occurs. On the other hand, many automated driving tasks operate in a closed-loop manner. That means a small amount of uncertainty induced by the AI-based S\&P may propagate over time and adversely impact the decision-making and planning systems of vehicles. However, it is still unclear how and to what degree the AI-driven S\&P processes can affect the closed-loop behavior of an AV. Those issues are fundamental for us to answer how safe and reliable the automated driving is \cite{ZHANG2,Piazzoni,Nascimento}.

Virtual environments and simulations are the most common ways to test an automated driving system (ADS), as they reduce the risk of property damage as much as possible, and many driving scenarios can be conveniently built up and safety metrics can be easily accessed\cite{shuoFeng,9564358,Junyao}. While economically efficient, this type of approaches is challenging to accurately reflect the true characteristics of an ADS under test, since  the modeling for both environments and sensors can have major effects on the effectiveness of the results. To this end, one research route attempts to upgrade the fidelity of the simulation to capture more details, while it may scarify the computational efficiency as the modeling refines \cite{jianZhao}. To overcome this issue,  recently, the so-called perception error model (PEM) is proposed. Rather than seeking a direct modeling of input-output relation for individual sensors (e.g., LiDAR, radar, and camera), PEM focuses on overall perception errors that incorporate the uncertainty from both sensing and AI-based signal processing. Such models are then used for testing the safety of AI-based AVs across different set of scenarios\cite{Piazzoni,Innes}.

Even though the simulation-based testing is in some sense economic and efficient, due to the potential mismatch to the real world, it is essentially restrictive to reflect real dynamics of an ADS. In addition, such testing is mainly useful in assessing the safety and performance of vehicles with given ADSs, it is normally hard for engineers to gain insights into ADS design problems, as there is no analytical results concluded that reveal the relationship between driving policy, perception systems, and performance of AVs. Recently, studies have been reported in addressing decision-making and planning within uncertain environments. Perception-aware methods are developed in order to consider together the effect of perception quality when planing actions \cite{Falanga,Lee,Salaris,FARINA201653}. These approaches almost rely on either accurate modeling for perception or additional design of estimators. To account for the uncertainty, chance-constrained model predictive controllers are proposed to handle probabilistic constraints that introduced by the imperfect perception \cite{Bonzanini1,BONZANINI}.  An information gain maximization approach is presented for  path planing in occluded areas  \cite{Gilhuly}. It actively explores the uncertain region and thus avoid overly conservativeness, as frequently incurred in traditional worst-case based methods. In \cite{zhai2021}, a delayed velocity feedback is addressed for a car-following task.  An intelligent driving model (IDM) based controller is studied for vehicle platooning, which considers the effect of noisy measurements \cite{li2021}. Research in  \cite{liu2022pnnuad,villenas} deals with the potential uncertainty, e.g., misclassificaiton and packet loss, during the sensing and perception stages. 

 Understanding the closed-loop behavior of AI-powered autonomous vehicles would be particularly relevant to ensure safe and reliable driving. While some initial efforts have been made to explore how deep learning based  S\&P systems affect AV outputs, it is still far from comprehensive knowledge; for example, key questions like whether it is robust or, furthermore, how robustly an AV can perform against the uncertainty coming from perception are not well addressed. In addition, existing works on controller synthesis are primarily focused on handling either Gaussian-distributed perception errors or other isolated uncertainty and, in particular, most of them lack a formal verification. These limitations greatly obstruct their applicability in practice. To fill up the research gap, this paper aims to develop analysis and synthesis tools for a class of ADSs that are affected by multiple sources of AI-induced S\&P uncertainty. The main contributions of the work are threefold:
\begin{itemize}
    \item Inspired by the work on PEM \cite{Piazzoni,Innes}, the effects of AI-based S\&P systems to the vehicles are described using the form of perception errors. In particular, three different types of error patterns are identified and modeled using Markov chains, Gaussian processes, and bounded disturbances. Based on that, we present a PEM-based automated driving model (PEM-ADM), which formalizes the impacts of AI-based S\&P systems to an ADS and enables rigorous analysis. It is worth noting that PEM-ADM extends existing control system models by explicitly including a heterogeneous source of uncertainty in the feedback loop.
    \item With the help of PEM-ADM, the closed-loop properties of  AI-based ADSs are studied. More specifically, the stochastic stability of the closed-loop is established by checking the feasibility of a set of linear matrix inequalities (LMIs). It demonstrates that an AI-based ADS may not even be stochastically stable (SS), if certain conditions are violated. This offers insights into how AI-based S\&P models can affect the reliability of ADSs and, in turn, guides SS-aware ADS design. Further, a bound on ADS steady-states is provided as well, which reveals the relationship between steady-state performance, control policy, and S\&P systems. The LMI-based conditions are then extended to deal with the problem of SS control synthesis.
    \item The robustness issue of AI-based ADSs is further discussed. A novel concept of stochastic guaranteed costs is first introduced to quantify how robustly an ADS can behave against the AI-induced uncertainty. Criteria are developed to test the ADS robustness. In addition, to design controllers that ensure a specified level of robustness, conditions are derived in the form of LMIs, which can be efficiently verified in practice. It is interesting to note that the stochastic optimal guaranteed cost control can be addressed straightforwardly in the proposed LMI framework, while this type of problem is usually hard to solve in other frameworks. The results obtained are then applied to a case study of car following to demonstrate their usefulness.
    
\end{itemize}


The rest of the paper is organized as follows: Section~\ref{sec2} presents the PEM-ADM and describes the problems. The analysis and synthesis for closed-loop properties of an AI-based  ADS are discussed in Section~\ref{sec3}. Robustness and performance in terms of stochastic guaranteed costs are studied in Section~\ref{sec4}. A car-following case study is addressed in Section~\ref{sec5}. Section~\ref{sec6} conducts the simulation and Section~\ref{sec7} concludes the paper.

\textit{Notation:} Denote $\mathbb{E}[\cdot]$ as the mathematical expectation of a random variable. $\| x\|$ stands for the 2-norm of a vector $x$. $A^T$  denotes the transpose of a matrix $A$. A positive definite  matrix is denoted by $A > 0$ and, conversely, $A < 0$ a negative definite matrix. $\lambda^{max}(A)$, $\lambda^{min}(A)$ denote the maximum and minimum eigenvalues of $A$, respectively. Denote $A^{-1}$  as the inverse of $A$. $\text{diag}(\cdot)$ defines a diagonal matrix. $\text{tr}(A)$ represents the trace of $A$. In a symmetric block matrix, $*$ denotes the corresponding symmetric counterpart, while in an optimization problem it stands for the optimum.

\section{modeling and problem description} \label{sec2}
\subsection{Perception Error Model}
\begin{figure}[!tbp]
\centering
\includegraphics[width=0.45\textwidth,trim=0.85cm 0.3cm 0.9cm 0.3cm, clip]{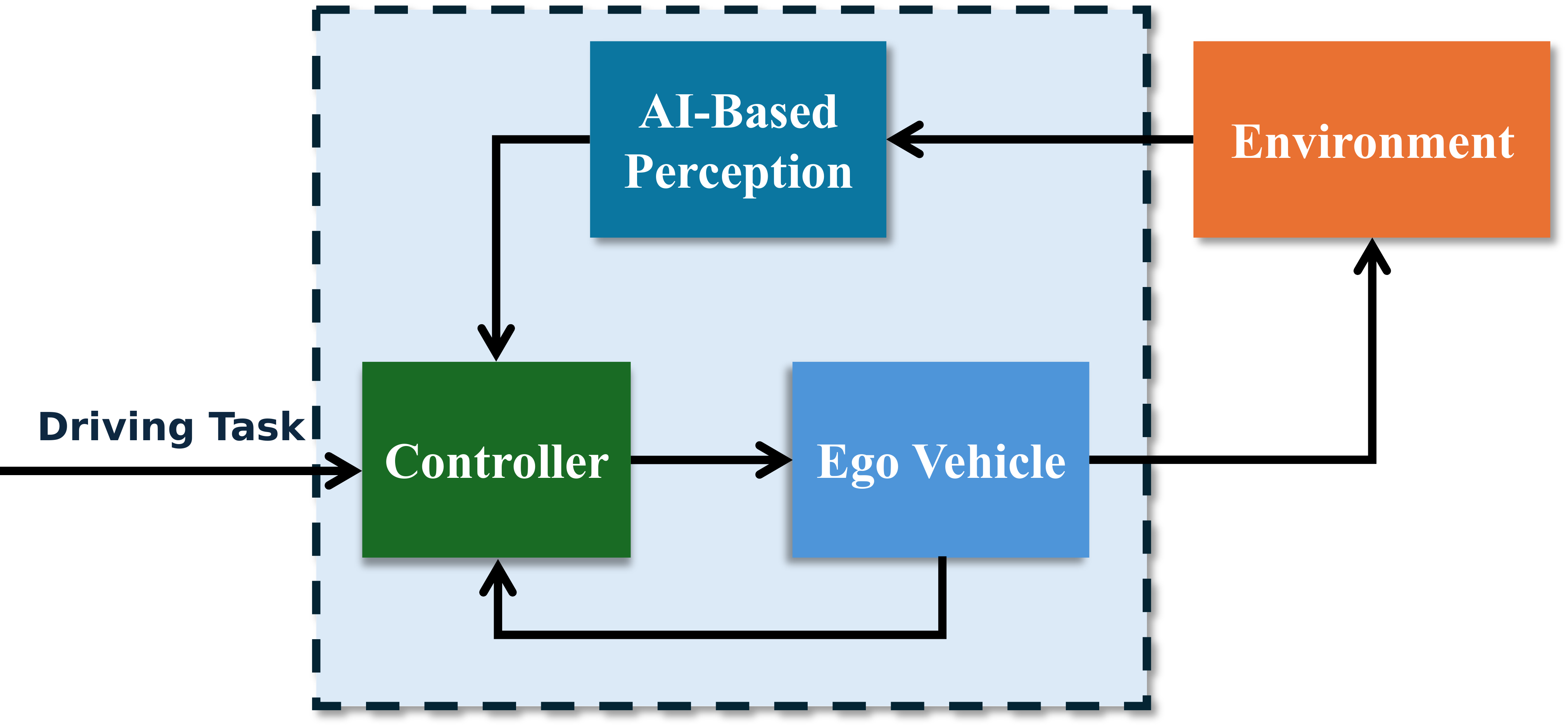}
\caption{Schematic of signal flows within an automated driving system.}
\label{fig:01}
\end{figure}

The signal flow schematic of a typical automated driving system (ADS)  is shown in Fig.~\ref{fig:01}. A key challenge in analyzing AI-based ADSs lies in the absence of a simple yet expressive way to represent the processes of sensing and perception. Recently, the perception error model (PEM) offers an approach to model AI-driven S\&P. Rather than directly modeling the input-output relationship of the S\&P, PEM focuses on the error characteristics induced by the S\&P. That is, PEM models the S\&P processes as
\begin{align}
    PEM(\mathcal{W}) = \mathcal{W}+ \mathcal{\varepsilon} \label{eq:1}
\end{align}
where $ \mathcal{W}$ denotes the environment and $\varepsilon$ captures the potential errors or uncertainty introduced by the S\&P processes. In this description, the focus shifts from modeling the entire S\&P pipeline to characterizing only the errors induced by it.

In this paper, we are mainly devoted to three important classes of  S\&P induced errors, that is, stochastic jumping, measurement noise, and bounded bias. These three have been demonstrated to be common and significant in practical S\&P systems and hamper the performance and safety of current AI-embedded AV systems \cite{Piazzoni}. In the following, we will present more detailed modeling on these typical perception errors.

\subsection{PEM-Based Automated Driving Modeling}
\begin{figure}[!tbp]
\centering
\includegraphics[width=0.45\textwidth,trim=0.85cm 0.3cm 0.9cm 0.3cm, clip]{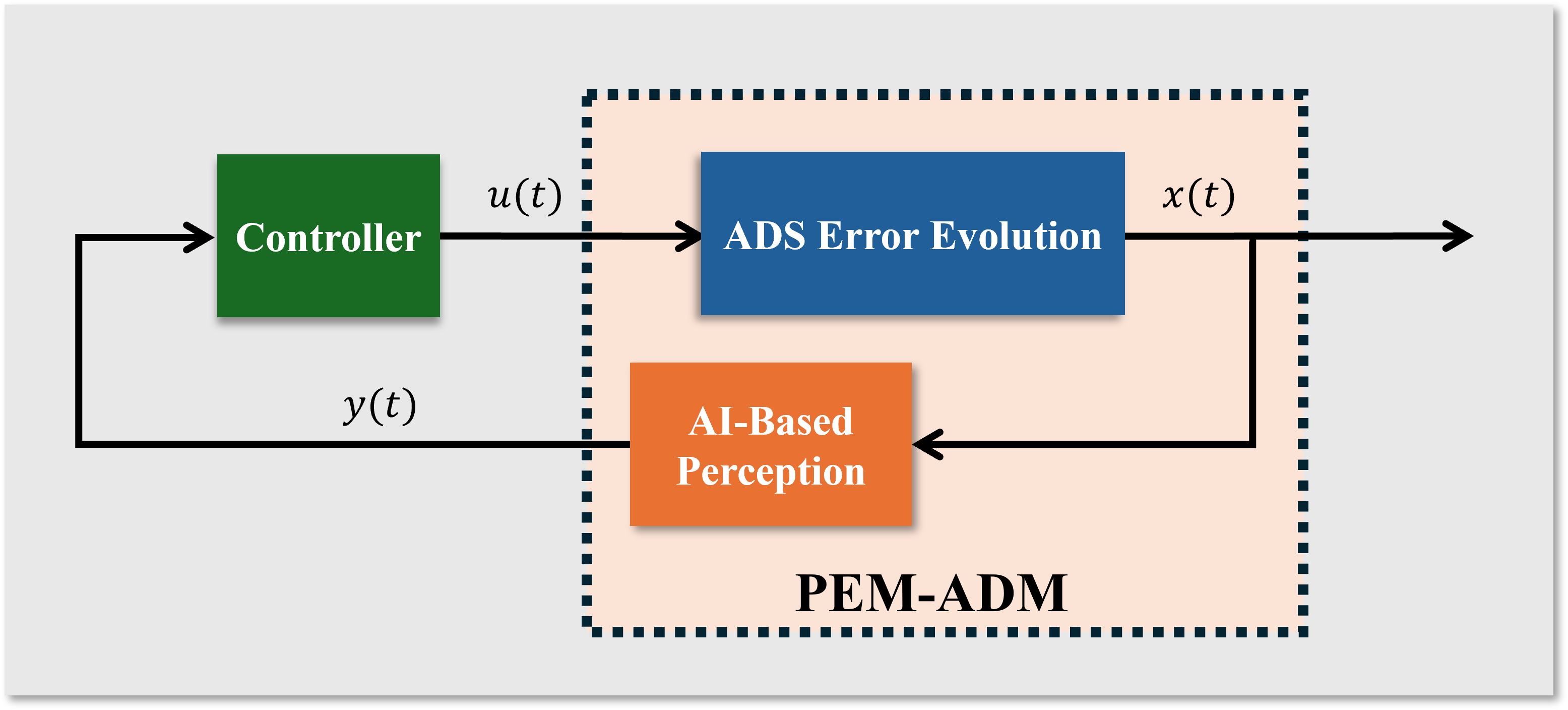}
\caption{Control block diagram for PEM-based automated driving model.}
\label{fig:02}
\end{figure}
Fig. \ref{fig:01} illustrates clearly how signals flow within an automated driving system, and this subsection introduces a control-oriented formulation which would make ADS analysis and synthesis easier. The dynamics of environments can be represented by
\begin{align}
    x^{env}(k+1) &= f^{env} \left( x^{env}(k) \right) \label{eq:2}
\end{align}
where $x^{env} $ denotes the state of the environment, which includes information like locations, poses, speeds, accelerations, etc., and $f^{env}(\cdot) $ governs the evolution of the state.

The motion of the ego vehicle is usually described by
\begin{align}
    x^{ego}(k+1) &= Ax^{ego}(k) + Bu (k) \label{eq:4}
\end{align}
with $x^{ego} $ being the state of the ego vehicle and $u$ the control input to be designed; $A$ and  $B$ are the system and input matrices, respectively.

The goal of ADSs can be typically formulated as maintaining certain synchronization with the environments, which contains tasks like lane keeping, car following, and cruise control. Without loss of generality, this paper considers that the objective of an ADS is to design  a control policy for $u$ to make $x^{ego}$ track $x^{env}$ over time. It should be noted also that, generally, the state dimensions of ego vehicles and environments do not necessarily match. However, the state of interest can be always tailored for specific driving tasks such that both are in the same space. Thus, this allows us to define the tracking error as $x = x^{ego}-x^{env} $, and  the following PEM-based automated driving model (PEM-ADM) is proposed to facilitate analysis and synthesis of AI-based ADSs:
\begin{align}
    &x (k+1) = Ax(k)  + B u(k)   \label{eq:5} \\
    &y (k) = C_i x(k) + D_i w(k) + E_i v(k) \nonumber \\
    & \quad i = r(t), \quad i\in \mathcal{I} = \left\{1,2,\cdots,N \right\} \label{eq:6}
\end{align}
where $x \in \mathbb{R}^{n_1} $, as already mentioned, indicates the deviation from the desired ones, and $u \in \mathbb{R}^{n_2} $ is the control input to be designed. $y\in \mathbb{R}^{n_3}$ denotes the perceived environment information produced by the AI-based S\&P systems, yet  corrupted by a heterogeneous sources of perception uncertainties; in particular, $w(k) \sim \mathcal{N}(0,1) $ denotes the i.i.d. standard Gaussian noise;  $v(k) \in \mathbb{R}^{n_3}$ with $\|v(k)\|$ being bounded accounts for the unknown low-frequency disturbance (e.g., sensor bias); $r(t) $ is a Markov chain defined by the transition probability matrix $\Xi=[p_{ij}] \in \mathbb{R}^{N\times N}$, which models some stochastic switching phenomena due to sensor failures, occlusions, and other algorithm failures. $C_i, D_i,E_i \in \mathbb{R}^{n_3\times n_3}$ are the observation, covariance, and bias matrices, respectively, and all of them can be influenced by stochastic switching.

\begin{myremark}
    The PEM-based automated driving model consists of two components: The difference equation \eqref{eq:5} describes the error evolution of an ADS, and \eqref{eq:6} is the measurement equation, which follows the PEM description \eqref{eq:1} to characterize AI-driven S\&P processes. It is important to note that the difference equation \eqref{eq:5} may not accurately capture the true error evolution of ADSs, as the behavior of environments can be much complicated, as indicated by \eqref{eq:2}, where $f^{env}(\cdot)$ represents any nonlinear function. In this work, the environment dynamics is purposefully simplified, since the main focus is to study the unique impacts caused by the AI-induced perception uncertainty. Actually, nonlinearities in system dynamics have been extensively studied in the nonlinear control literature.
\end{myremark}
\begin{myremark}
    The AI-based ADSs described by \eqref{eq:5} and \eqref{eq:6} are significantly different from traditional control system structures, where uncertainties typically appear in system equations while perfect state measurements are often assumed. In contrast, the AI-based systems may involve more diverse and complex sources of uncertainties in perception. This fundamental difference introduces new challenges for system analysis and synthesis.
\end{myremark} 
\subsection{Problem Description}
A control block diagram for the proposed PEM-ADM is illustrated in Fig. \ref{fig:02}. In order to analyze and design the closed-loop properties of ADSs, this paper considers the following control policy that uses the perceived environment information as the input
\begin{align}
    u(k) = K_i y(k), \qquad i \in \mathcal{I} \label{eq:7}
\end{align}
where $K_i \in \mathbb{R}^{n_2\times n_3}$ are the control gains subject also to the Markov switching. 

Plugging in \eqref{eq:7} into \eqref{eq:5} gives the closed-loop system:
\begin{align}
    x(k+1) &= (A+BK_iC_i )x(k) + B K_i  D_iw(k) + B K_i E_i v(k) \nonumber\\
    &  \stackrel{\text{def}}{=} A^{cl}_i x(k) + D^{cl}_i w(k) + E^{cl}_i v(k), \quad i \in \mathcal{I} \label{eq:8}
\end{align}
It is clear that due to the effects of unreliable AI-based perception, the closed-loop system~\eqref{eq:8} is subject to a heterogeneous sources of uncertainties which include stochastic switching, noise, and bias and affect AV's behavior convolutedly.

Before approaching it, the following concepts for stochastic stability (SS) properties are defined.
\begin{mydefinition}
    The solution $x(k)\equiv 0$ of \eqref{eq:8} is said to be uniformly bounded in mean square if there exists $M_0>0$ such that
\begin{align}
    \mathbb{E}[x(k)^2] \le M_0, \quad \forall ~ k \label{eq:9}
\end{align}
In particular, $x(k)$ is said to be uniformly ultimately bounded in mean square, if there exist $M_1, K>0$ such that 
\begin{align}
    \forall ~ k \ge K \quad \Rightarrow \quad \mathbb{E}[x(k)^2] \le M_1 \label{eq:10}
\end{align}
Furthermore, if $M_1 \to 0$ as $k \to \infty$, then we say  $x(k)$ is uniformly asymptotically stable in mean square.
\end{mydefinition}

The first goal of this paper is to develop a test for the stochastic stability of the closed-loop system~\eqref{eq:8} and provide an efficient method to synthesize stabilizing controllers. The second goal will further discuss how robust an ADS can be when subject to multiple sources of AI-induced uncertainties $r(k)$, $w(k)$ and $v(k)$. The third objective is to solve a class of stochastic optimal control problems associated with the system \eqref{eq:5}, \eqref{eq:6} and controller \eqref{eq:7}.



\section{ Stochastic Stabilizing Control} \label{sec3}
This section addresses the stochastic stability analysis and control synthesis for a class of AI-based ADSs \eqref{eq:5}, \eqref{eq:6} under the control structure~\eqref{eq:7}. 

\subsection{Closed-loop Stability Analysis}
The following lemmas are introduced to assist in the derivation of main results.
\begin{mylemma}\label{lem1}
    Consider a stochastic processes $(z_1(k),z_2(k))$ satisfying the Markov property. Let $z_{1,k} = z_1(k)$ and $z_{2,k} = z_2(k)$ be the values the processes take at each time $k$. Then, $z_1(k)$ is  uniformly ultimately bounded in mean square, if there exits a positive semidefinite function $W(z_1,z_2)$ such that for some  $c_1,c_3>0$, and $c_2 \in(0,1)$
    \begin{align}
        &c_1\| z_1\|^2  \le  W(z_1,z_2), \qquad \forall z_1,z_2 \label{eq:11}\\
            \mathbb{E}[W(z_1(k+1),&z_2(k+1))|z_{1,k},z_{2,k}] - W(z_{1,k},z_{2,k}) \nonumber  \\
            & \le - c_2   W(z_{1,k},z_{2,k}) + c_3 , ~ \forall z_{1,k},z_{2,k} \label{eq:12}
    \end{align}
\end{mylemma}
\begin{proof}
      By \eqref{eq:12}, it is equivalent that for any $(z_{1,k-1},z_{2,k-1})$ the following is true
    \begin{align}
        &\mathbb{E}[W(z_1(k),z_2(k))|z_{1,k-1},z_{2,k-1}] \nonumber \\
        & \qquad \qquad \le (1-c_2) W(z_{1,k-1},z_{2,k-1}) + c_3 \label{eq:13}
    \end{align}
    The stochastic processes $(z_1(k),z_2(k))$ satisfy the Markov property and, hence, 
    \begin{align}
        &\mathbb{E}[W(z_1(k),z_2(k))|z_{1,k-2},z_{2,k-2}] \nonumber \\ 
        & = \mathbb{E}[\mathbb{E} [W(z_1(k),z_2(k))|z_1(k-1),z_2(k-1)]|z_{1,k-2},z_{2,k-2}] \label{eq:14}
    \end{align}
    Using \eqref{eq:13}, it follows that
    \begin{align}
         & \mathbb{E}[W(z_1(k),z_2(k))|z_{1,k-2},z_{2,k-2}] \nonumber \\ 
         & ~ \le (1-c_2)\mathbb{E}[   W(z_1(k-1),z_2(k-1)) | z_{1,k-2},z_{2,k-2}  ] + c_3 \nonumber \\
         & ~ \le  (1-c_2)^2 W(z_{1,k-2},z_{2,k-2})+ (1-c_2) c_3 +c_3 
    \end{align}
    Again, applying \eqref{eq:13} and \eqref{eq:14}  yields
    \begin{align}
        & \mathbb{E}[W(z_1(k),z_2(k))|z_{1,k-3},z_{2,k-3}] \nonumber \\
         &~ = \mathbb{E}[\mathbb{E} [W(z_1(k),z_2(k))|z_1({k-2}),z_2({k-2})]|z_{1,k-3},z_{2,k-3}]  \nonumber \\
         &~ \le (1-c_2)^3 W(z_{1,k-3},z_{2,k-3}) +(1-c_2)^2 c_3 \nonumber \\
          &\quad+(1-c_2) c_3 +c_3
         \label{eq:16}
    \end{align}
    Iterating this process gives
    \begin{align}
         & \mathbb{E}[W(z_1(k),z_2(k))|z_{1,0},z_{2,0}] \nonumber \\
         & ~= \mathbb{E}[\mathbb{E} [W(z_1(k),z_2(k))|z_1({1}),z_2({1})]|z_{1,0},z_{2,0}]  \nonumber \\
         & ~\le \quad \cdots \cdots  \nonumber \\
         &  ~\le (1-c_2)^k W(z_{1,0},z_{2,0}) + \sum _{n=0}^{k-1} (1-c_2)^n c_3 \nonumber \\
         & ~ = (1-c_2)^k W(z_{1,0},z_{2,0}) + \frac{c_3 [1-(1-c_2)^k]}{c_2}\label{eq:17}
        \end{align}
    In view of \eqref{eq:11},
    \begin{align}
        c_1 \mathbb{E}[\|z_1(k)\|^2 | z_{1,0},z_{2,0}] \le \mathbb{E}[W(z_1(k),z_2(k))| z_{1,0},z_{2,0}]  \label{eq:18}
    \end{align}
    Together with \eqref{eq:17}, it follows that
    \begin{align}
         \mathbb{E}[\|z_1(k)\|^2 |  z_{1,0},z_{2,0}] &\le  \frac{(1-c_2)^k}{c_1} W(z_{1,0},z_{2,0}) \nonumber \\
         & \quad+ \frac{c_3 [1-(1-c_2)^{k}]}{c_1c_2}
    \end{align}
    Let $k \to \infty$. We get
    \begin{align}
        \mathbb{E}[\|z_1(\infty)\|^2 ] \le \frac{c_3}{c_1c_2}
    \end{align}
    This implies the required result, which completes the proof.
\end{proof}

 \begin{mylemma}[Discrete Dynkin's Formula] \label{lem2}
    Let $(z_1(k),z_2(k))$ be the Markov processes and $W(z_1,z_2)$ be a scalar-valued function. For any $k =1,2,\cdots$ the following equality holds
        \begin{align}
           \mathbb{E}\left[\sum_{s=1}^k \mathbb{E}[W(z_1(s),z_2(s)) |z_{1,s-1},z_{2,s-1}] \right. \nonumber\\
           \left.-W(z_{1,s-1},z_{2,s-1}) \biggl | z_{1,0},z_{2,0} \right]  \nonumber \\
           = \mathbb{E}[W(z_1(k),z_2(k)) |z_{1,0},z_{2,0}] -  W(z_{1,0},z_{2,0}) \label{eq:21}
        \end{align}
    \end{mylemma}

The following theorem establishes the stochastic stability of AI-based ADSs in terms of the feasibility of a linear matrix inequality (LMI).
\begin{mytheorem} \label{thm1}
The solution $x(k) \equiv0$ of the closed-loop system~\eqref{eq:8} is uniformly ultimately bounded in mean square, if there exist $P_i >0$  ($i \in \mathcal{I}$) such that the following LMI holds
    \begin{align}
        \sum_{j\in \mathcal{I}} p_{ij} (A_i^{cl})^TP_jA_i^{cl}  -  P_i < 0, \qquad i \in \mathcal{I} \label{eq:22}
    \end{align}
\end{mytheorem}

\begin{proof}
    Consider the following switched Lyapunov function candidate
\begin{align}
    V\left(x(k),r(k)\right) = x(k)^T P_ix(k) \label{eq:23}
\end{align}
where $P_i >0$ indexed by $i=r(k)$ which is governed by a Markov chain with probability transition matrix $\Xi = [p_{ij}]$. The forward difference of \eqref{eq:23} is defined as
\begin{align}
    \Delta V(k) = V(x(k+1),{r(k+1)})-V(x(k),{r(k)}) \label{eq:24}
\end{align}

Given $x(k)=x_k, r(k)=i$ at time $k$, the conditional expectation of \eqref{eq:24} can be calculated as
\begin{align}
    & \mathbb{E}[\Delta V(k)|x_k,i] =  \mathbb{E}[V(x(k+1),{r(k+1)})|x_k,i]-V(x_k,i)  \label{eq:25}
\end{align}
Let $j = r(k+1)$. By \eqref{eq:23}, together with \eqref{eq:8}, it yields
\begin{align}
    & \mathbb{E}[\Delta V(k)|x_k,i] = \mathbb{E}[V(x(k+1),j)|x_k,i]- x_k^T P_ix_k \nonumber \\
    &= \mathbb{E}[(A_i^{cl}x_k +D_i^{cl} w(k) + E_i^{cl} v(k))^T P_j \nonumber \\
    & \qquad \times(A_i^{cl}x_k +D_i^{cl} w(k) + E_i^{cl} v(k))] 
    - x_k^T P_ix_k
    \label{eq:26} 
\end{align}

    Notice that provided $x_k$ and $i$, $V(x(k+1),j)$ is a function of random variables $w(k)$ and $j$. Hence, by the properties of Gaussian processes and Markov chains, \eqref{eq:26} can be further calculated as
    \begin{align}
        & \mathbb{E}[\Delta V(k)|x_k,i] \nonumber \\
        &=\mathbb{E}[x_k^T(A_i^{cl})^TP_jA_i^{cl}x_k] +\mathbb{E}[x_k^T(A_i^{cl})^TP_jD_i^{cl}w(k)] \nonumber \\
        &~+\mathbb{E}[x_k^T(A_i^{cl})^TP_j E_i^{cl}v(k)] + \mathbb{E}[w(k)^T (D^{cl}_i)^TP_j A_i^{cl}x_k] \nonumber \\
        &~+  \mathbb{E}[v(k)^T (E^{cl}_i)^TP_j A_i^{cl}x_k] + \mathbb{E}[w(k)^T (D^{cl}_i)^TP_j D_i^{cl} w(k) ]  \nonumber \\
        &~+ \mathbb{E}[w(k)^T (D^{cl}_i)^TP_j E_i^{cl} v(k) ] +  \mathbb{E}[v(k)^T (E^{cl}_i)^TP_j  D_i^{cl} w(k)]  \nonumber \\
        &~+ \mathbb{E}[v(k)^T (E^{cl}_i)^TP_j  E_i^{cl} v(k)] - x_k^T P_ix_k\nonumber \\
        &= \sum_{j\in \mathcal{I}} p_{ij} x_k^T(A_i^{cl})^TP_jA_i^{cl}x_k+ \sum_{j\in \mathcal{I}} p_{ij} x_k^T(A_i^{cl})^TP_j E_i^{cl}v(k) \nonumber \\
        & ~+ \sum_{j\in \mathcal{I}} p_{ij} v(k)^T(E_i^{cl})^TP_jA_i^{cl}x_k + tr\left( \sum_{j\in \mathcal{I}} p_{ij}(D^{cl}_i)^TP_j D_i^{cl}\right) \nonumber \\
        &~ +\sum_{j\in \mathcal{I}} p_{ij}v(k)^T (E^{cl}_i)^TP_j  E_i^{cl} v(k) - x_k^T P_ix_k  
    \end{align}
    Applying Young's inequality, it follows that for any $\varepsilon_1>0$
    \begin{align}
        \mathbb{E}[\Delta V(k)|x_k,&i]  \le  \sum_{j\in \mathcal{I}} p_{ij} x_k^T(A_i^{cl})^TP_jA_i^{cl}x_k\nonumber\\
        & ~+ \varepsilon^{-1} _1 \sum_{j\in \mathcal{I}} p_{ij} x_k^T(A_i^{cl})^TP_j E_i^{cl}(E_i^{cl})^TP_jA_i^{cl}x_k \nonumber \\
        & ~ + \varepsilon _1 v(k)^Tv(k) + tr\left( \sum_{j\in \mathcal{I}} p_{ij}(D^{cl}_i)^TP_j D_i^{cl}\right) \nonumber \\
        &~ +\sum_{j\in \mathcal{I}} p_{ij}v(k)^T (E^{cl}_i)^TP_j  E_i^{cl} v(k) - x_k^T P_ix_k  \label{eq:28}
    \end{align}

    By the condition of the theorem, it implies that there is $\alpha_1>0$  such that for any $i \in \mathcal{I}$
    \begin{align}
        \sum_{j\in \mathcal{I}} p_{ij} x_k^T(A_i^{cl})^TP_jA_i^{cl}x_k  - x_k^T P_ix_k \le -\alpha_1 x^T_k x_k \label{eq:29}
    \end{align}
    Clearly, there is $\alpha_2 >0 $ such that for any $i \in \mathcal{I}$
    \begin{align}
        \sum_{j\in \mathcal{I}} p_{ij} x_k^T(A_i^{cl})^TP_j E_i^{cl}(E_i^{cl})^TP_jA_i^{cl}x_k \le \alpha_2 x^T_k x_k  \label{eq:30}
    \end{align}
   Therefore, we can let $\varepsilon_1$ sufficiently large so that 
    \begin{align}
        \alpha_3 \stackrel{\text{def}}{=} -\alpha_1 + \varepsilon^{-1}_1 \alpha_2 < 0 \label{eq:31}
    \end{align}
    On the other hand, in view of the boundedness of $v(k)$ it is easy to see that
    \begin{align}
        &\varepsilon _1 v(k)^Tv(k) + tr\left( \sum_{j\in \mathcal{I}} p_{ij}(D^{cl}_i)^TP_j D_i^{cl}\right) \nonumber \\
        &\quad + \sum_{j\in \mathcal{I}} p_{ij}v(k)^T  (E^{cl}_i)^TP_j  E_i^{cl} v(k) \le \alpha_4, \quad \forall  i \in \mathcal{I} \label{eq:32}
    \end{align}
    for some $\alpha_4 >0$.
    Hence, using \eqref{eq:29}--\eqref{eq:32}, it follows that
    \begin{align}
        \mathbb{E}[\Delta V(k)|x_k,i] \le -\alpha _3 x_k^Tx_k + \alpha_4 \label{eq:33}
    \end{align}
    Since $x^TP_ix >0 , \forall i \in \mathcal{I} $, we have
    \begin{align}
           x^TP_ix \le \lambda^{max}(P_i) x^T x
    \end{align}
    Therefore, \eqref{eq:33} is bounded by
    \begin{align}
        \mathbb{E}[\Delta V(k)|x_k,i] &\le -\alpha _3 (\lambda^{max}(P_i))^{-1} V(x_k,i) + \alpha_4 \nonumber \\
        & \le -\alpha_5 V(x_k,i)+ \alpha_4, \quad \forall  x_k,i \label{eq:35}
    \end{align}
    with $\alpha_5 =\alpha _3 \min\{ (\lambda^{max}(P_i))^{-1} \}$.
    By Lemma 1, this yields the required result.\end{proof}
    
\begin{myremark}
    Theorem \ref{thm1} reveals clearly the relationship between the closed-loop stability and transition probabilities~($p_{ij}$). It means that an AI-driven ADS may not be stochastically stable, if the condition \eqref{eq:22} is violated. In addition to that, through the bound \eqref{eq:35} obtained in the proof (particularly $\alpha_4$ as introduced in \eqref{eq:32}), it also suggests that the intensity of the noise and the bias in the sensing and perception systems can have significant impacts on the steady-state accuracy of an automated driving system.
\end{myremark}

\subsection{Stochastic Stabilizing Control Synthesis}
Theorem \ref{thm1} discusses the closed-loop properties of an AI-based AV system with a given controller.  This subsection shall be dedicated to control synthesis to help find candidate controllers that certify the desired stochastic stability. The following theorem treats this issue.

\begin{mytheorem} \label{thm:2}
    Consider the AI-based ADSs \eqref{eq:5}, \eqref{eq:6} with the control law \eqref{eq:7}. The closed-loop system is stochastically stabilizable, if there exist $S_i, Y_i  >0$ and $W_i$ such that the following LMIs hold for all $i \in \mathcal{I}$
    \begin{align}
    &\begin{bmatrix}
        -S_i &   [M_i(AS_i + BW_iC_i)]^T \\
        * &  -\Lambda  
    \end{bmatrix} <0 \nonumber \\
    &C_i S_i = Y_i C_i\label{eq:36}
\end{align}
with 
\begin{align*}
    M_i &= [\sqrt{p_{i1} }I, \cdots,\sqrt{p_{iN} }I ]^T\\
    \Lambda &= \text{diag}(S_1, \cdots,S_N)
\end{align*}
In particular, one admissible stochastic stabilizing control (SSC) gain is solved by $K_i=W_i Y_i^{-1}$.
\end{mytheorem}
\begin{proof}
It is known by Theorem~\ref{thm1}  that the closed-loop system is stochastically stable if the condition  \eqref{eq:22} is satisfied. By the definition of $A_i^{cl}$ as given in \eqref{eq:8},  test \eqref{eq:22} is equivalent to
\begin{align}
      \sum_{j\in \mathcal{I}} p_{ij} (A+BK_iC_i)^TP_j(A+BK_iC_i)  -  P_i < 0, \quad \forall i\in \mathcal{I} \label{eq:37}
\end{align}
Let $P_i = S_i^{-1}$. Due to the fact that a congruent transformation preserves definiteness, we pre- and post-multiply \eqref{eq:37} by $S_i$ and obtain equivalently that for any $i\in \mathcal{I}$
\begin{align}
      S_i (A+BK_iC_i)^T (\sum_{j\in \mathcal{I}} p_{ij}S_j^{-1})(A+BK_iC_i) S_i  - S_i < 0 \label{eq:38}
\end{align}
Notice that $\sum_{j\in \mathcal{I}} p_{ij}S_j^{-1} = M_i^T \Lambda^{-1} M_i$ and  $S_i = S_i ^T$. Hence, by applying Schur complement, \eqref{eq:38} amounts to
\begin{align}
    \begin{bmatrix}
        -S_i & S_i^T (A +B K_i C_i)^TM_i^T \\
         * & - \Lambda
    \end{bmatrix} <0, \quad \forall i\in \mathcal{I}  \label{eq:39}
\end{align}
Using $ W_i= K_i Y_i$, it is not difficult to check that the condition~\eqref{eq:36} implies \eqref{eq:39}. This completes the proof.\end{proof}

\section{Stochastic Guaranteed Cost Control}\label{sec4}
The previous discussion is mainly central to the stochastic stability of an AI-driven ADS. While the stability plays a fundamental role in various AV applications, the performance validation is also critical for an ADS to perform satisfactorily, especially when AI-induced uncertainties are present. In addition, it is always of practical interest to study optimal synthesis in the presence of uncertainties. For this purpose, this section further discusses the performance and robustness of an AI-based autonomous driving system and presents a stochastic optimal control framework to reach the best guaranteed performance despite uncertainties.

\subsection{Stochastic Guaranteed Cost}
The following stochastic cost function is used to measure the performance of an AI-based ADS, which is common in most stochastic optimal control literature:
\begin{align}
    J = \mathbb{E} \left[\sum _{k=0}^{N} x(k)^T Q x(k) + u(k)^T R u(k) \biggm| x_0,r_0\right] \label{eq:41}
\end{align}
where $(x_0,r_0)$ is the initial state, $N =1,2,\cdots$ is a time horizon, and $Q,R>0$ are weighted matrices with appropriate dimensions.

It is usually challenging to directly optimize \eqref{eq:41}, since the control $u$ in our cases is based on uncertain feedback, containing multiple sources of errors. As a result, the cost $J$ will be convolutedly correlated with both control and various perception uncertainties. This makes optimization intractable. To overcome it, one natural way is to construct a nice upper bound for the cost function and look for an optimal guaranteed cost.

The following definition formalizes the above idea.
 \begin{mydefinition}
        It is said that an uncertain system associated with the cost function \eqref{eq:41} has a stochastic guaranteed cost with respect to the uncertainties $r(k)$, $w(k)$, and $v(k)$, if there exist some common constants $\gamma,M_2 >0$ such that $\forall N, r,w,v $ we have
        \begin{align}
             J &= \mathbb{E} \left[\sum _{k=0}^{N} x(k)^T Q x(k) + u(k)^T R u(k) \biggm| x_0,r_0\right] \nonumber \\
             &\le \gamma ^2 \mathbb{E}\left[ \sum _{k=0}^{N} \begin{bmatrix}
             w(k)^T&v(k)^T\end{bmatrix} \begin{bmatrix}
                 w(k)\\
                 v(k)
             \end{bmatrix}\biggm| x_0,r_0\right] + M_2\nonumber \\ 
             &  \stackrel{\text{def}}{=}\bar{J}
             \label{eq:42}
        \end{align}
        In particular, $\bar{J}$ is said to be a stochastic $\gamma$-guaranteed cost, $u(k)$  be a stochastic $\gamma$-guaranteed cost controller, where $\gamma $ quantifies the robustness level of the system against uncertainties.
    \end{mydefinition}

The next theorem establishes the existence of a stochastic guaranteed cost for an AI-based ADS\eqref{eq:5}, \eqref{eq:6}.

\begin{mytheorem} \label{thm3}
    Consider the closed-loop system \eqref{eq:8} generated by the control law \eqref{eq:7}, associated with the cost function \eqref{eq:41}. It has a stochastic guaranteed cost with respect to the uncertainties $r(k)$, $w(k)$, and $v(k)$, if there exist $P_i>0$ and $\gamma >0$ such that the following LMI holds
    \begin{align}
    \begin{bmatrix}
            \Pi_{11} & \Pi_{12} & 0  \\
             * & \Pi_{22} & 0\\
             * & * & \Pi_{33} \\
        \end{bmatrix} < 0 \label{eq:43}    
    \end{align}
    \vskip -10pt
    \begin{align*}
        \Pi_{11} &= \sum_{j\in \mathcal{I}} p_{ij}(A_i^{cl})^TP_jA_i^{cl} -P_i + Q +(K_iC_i)^T R K_iC_i \\
        \Pi_{12} &=  \sum_{j\in \mathcal{I}} p_{ij}(A_i^{cl})^TP_jE_i^{cl} + (K_iC_i)^T R K_i E_i \\
        \Pi_{22} & =   \sum_{j\in \mathcal{I}}p_{ij}(E_i^{cl})^TP_jE_i^{cl} + (K_iE_i)^T RK_iE_i -\gamma^2 I \\
        \quad \Pi_{33} &=    \sum_{j\in \mathcal{I}}p_{ij}(D_i^{cl})^TP_jD_i^{cl}+(K_iD_i)^T R K_i D_i-\gamma^2I
    \end{align*}
\end{mytheorem}

\begin{proof}
With the control law \eqref{eq:7}, the cost function \eqref{eq:41} can be rewritten as
\begin{align}
    J = \mathbb{E} \left[\sum _{k=0}^{N} x(k)^T Q x(k) + (K_iy(k))^T R K_iy(k) \biggm| x_0,r_0\right] \label{eq:44}    
\end{align}
Equivalently, our goal is to construct a stochastic guaranteed cost for \eqref{eq:44}. 

Let $\xi = [w^T, v^T]^T$. We claim:
\begin{align}
    \mathbb{E}[\Delta V(k)|x_k,i]  &\le - x_k^TQx_k -  \mathbb{E}[(K_iy_k)^T R K_iy_k|x_k,i]\nonumber \\
    & \qquad \qquad+ \gamma^2 \mathbb{E}[\xi_k ^T \xi_k |x_k,i] \label{eq:45}
\end{align}
where $\Delta V$ is exactly previously defined referring to \eqref{eq:24}.
\textit{ Proof of claim.} Substituting \eqref{eq:6} and \eqref{eq:26} into \eqref{eq:45}, it follows that 
\begin{align}
    &\sum_{j\in \mathcal{I}} p_{ij} x_k^T(A_i^{cl})^TP_jA_i^{cl}x_k+ \sum_{j\in \mathcal{I}} p_{ij} x_k^T(A_i^{cl})^TP_j E_i^{cl}v(k) \nonumber \\
        & + \sum_{j\in \mathcal{I}} p_{ij} v(k)^T(E_i^{cl})^TP_jA_i^{cl}x_k + tr\left( \sum_{j\in \mathcal{I}} p_{ij}(D^{cl}_i)^TP_j D_i^{cl}\right) \nonumber \\
        & +\sum_{j\in \mathcal{I}} p_{ij}v(k)^T (E^{cl}_i)^TP_j  E_i^{cl} v(k) - x_k^T P_ix_k \le - x_k^TQx_k \nonumber \\
        & -(K_i C_i x_k)^T RK_i C_i x_k  -(K_i C_i x_k)^T R K_i E_i v(k) \nonumber \\
        &- (K_i E_i v(k))^T RK_i C_i x_k- tr\left((K_i D_i )^T RK_i D_i  \right)\nonumber \\
        & -(K_i E_i v(k))^T RK_i E_i v(k) +\gamma^2  + \gamma^2 v(k)^Tv(k) \label{eq:46}
\end{align}
Let $\bar{w}$ such that $\bar{w} ^T \bar{w} =1$ . Thus, 
\begin{align}
    &tr\left( \sum_{j\in \mathcal{I}} p_{ij}(D^{cl}_i)^TP_j D_i^{cl}\right) = \sum_{j\in \mathcal{I}} p_{ij} \bar{w}(k)^T(D^{cl}_i)^TP_j D_i^{cl}\bar{w}(k) \nonumber \\
    &tr\left((K_i D_i )^T RK_i D_i  \right) = (K_i D_i \bar{w}(k))^T RK_i D_i\bar{w}(k) \nonumber \\  
    &\gamma ^2 = \gamma ^2 \bar w(k)^T \bar w(k) \label{eq:47}
\end{align}
With \eqref{eq:47}, it is easy to verify that \eqref{eq:43} implies \eqref{eq:46}. This proves the claim. \qed

Taking expectation and applying Lemma \ref{lem2} to \eqref{eq:45}, it yields that
\begin{align}
    &\mathbb{E}[V(x(N),r(N))| x_0,r_0] - V(x_0,r_0) \nonumber \\
    & \qquad \le  \mathbb{E} \left[\sum _{k=0}^{N} -x(k)^T Q x(k) - (K_iy(k))^T R K_iy(k) \right.\nonumber \\
     & \qquad \qquad \qquad \left.+ \gamma^2 \xi_k ^T \xi_k  \biggm| x_0,r_0\right] 
\end{align}
By definition \eqref{eq:23}, $V$ is positive definite, and hence,
\begin{align}
    \mathbb{E} \left[\sum _{k=0}^{N} x(k)^\top  Q x(k) + (K_iy(k))^T R K_iy(k) \biggm| x_0,r_0\right]  \nonumber \\
    \le \gamma ^2 \mathbb{E} \left[\sum _{k=0}^{N} \xi_k ^T \xi_k  \biggm| x_0,r_0\right]  + V(x_0,r_0)
\end{align}
That is, a stochastic guaranteed cost is constructed, which completes the proof.\end{proof}

\begin{myremark}
    It is not difficult to see that the condition \eqref{eq:43} implies the condition \eqref{eq:22} when looking at $\Pi _{11}$, which in turn ensures the stochastic stability in light of Theorem 1.
\end{myremark}

\begin{myremark}
    Theorem \ref{thm3} can be used to test how robust an ADS can be to the AI-induced uncertainties in the sense of \eqref{eq:42} in terms of the metric $\gamma$. For example, we can simply minimize $\gamma$ subject to the LMI constraint \eqref{eq:43}.
\end{myremark}

\subsection{Stochastic Optimal Guaranteed Cost Control}
In practice, we not only would like to test a given system, but also design an ADS with an optimal guaranteed cost. This subsection aims to extend Theorem \ref{thm3} to address the problem of stochastic optimal control synthesis. 
\begin{figure*}[!t]
   \begin{align}
        &        \begin{bmatrix}
        -S_i&  0 & 0 &0  & (W_iC_i)^T  & 0 & [M_i(AS_i +BW_iC_i)]^T  & S_i\\
        * &  -  \gamma^2 \lambda^2  I & 0 & 0 & (W_iE_i)^T &0 & ( M_iBW_iE_i)^T & 0\\
        * & * & -  \gamma^2 \lambda^2  I & (W_iD_i)^T& 0 & (M_iBW_iD_i)^T & 0  & 0\\
        * & * & * & -R^{-1} & 0 & 0 & 0 & 0\\
        * & * & * & * & -R^{-1} & 0 & 0  & 0\\
        * & * & * & * & * & -\Lambda & 0 & 0\\
        * & * & * & * & * & * & -\Lambda  & 0\\
        * & * & * & * & * & * & *  & -Q^{-1}
    \end{bmatrix} <0 \label{eq:50} \\
    & S_i > \lambda I, \quad C_iS_i = Y_i C_i, \quad D_iS_i = Y_i D_i, \quad E_iS_i = Y_i E_i  \label{eq:51}
    \end{align}
    \end{figure*}
\begin{mytheorem} \label{thm4}
    Consider system \eqref{eq:5}, \eqref{eq:6} with the control law \eqref{eq:7}. There is a stochastic guaranteed cost controller in the sense of \eqref{eq:42}, if for a given $\gamma$, there exist matrices $ S_i,  Y_i >0$, $W_i$ with compatible dimensions, and $\lambda>0$ such that the linear matrix inequalities \eqref{eq:50}, \eqref{eq:51} hold for all $i \in \mathcal{I}$, with $M_i$ and $\Lambda$ being defined as
\begin{align*}
    M_i &= [\sqrt{p_{i1} }I, \cdots,\sqrt{p_{iN} }I ]^T\\
    \Lambda &= \text{diag}(S_1, \cdots,S_N)
\end{align*}
    In particular, a stochastic $\gamma$-guaranteed cost control gain is solved as $K_i= W_i Y_i^{-1}$.
\end{mytheorem}
\begin{proof}
    To apply Theorem \ref{thm3}, indeed, we aim to prove the condition \eqref{eq:43}. Let $P_i = S_i^{-1}$.  \eqref{eq:43} can be equivalently written into
\begin{align}
        \begin{bmatrix}
            \Gamma_{11} & \Gamma_{12} & 0\\
            * & \Gamma_{22} & 0 \\
            * & *  &  \Gamma_{33} 
        \end{bmatrix} + \begin{bmatrix}
            0 & (K_iC_i)^T \\
            0 & (K_iE_i)^T \\
            (K_iD_i)^T & 0
        \end{bmatrix} \begin{bmatrix}
            R & 0\\
            0 & R
        \end{bmatrix}\nonumber \\
         \times \begin{bmatrix}
            0 & 0 & K_iD_i \\
            K_iC_i & K_iE_i & 0 
        \end{bmatrix} \label{eq:52}
    \end{align}
with
\begin{align*}
    \Gamma_{11} &= \sum_{j\in \mathcal{I}} p_{ij}(A_i^{cl})^TS^{-1}_jA_i^{cl} - S^{-1}_i + Q \\
    \Gamma_{12} &= \sum_{j\in \mathcal{I}} p_{ij}(A_i^{cl})^TS^{-1}_jE_i^{cl} \\
    \Gamma_{22} &=  \sum_{j\in \mathcal{I}}p_{ij}(E_i^{cl})^TS^{-1}_jE_i^{cl} -\gamma ^2 I\\
    \Gamma_{33} &= \sum_{j\in \mathcal{I}}p_{ij}(D_i^{cl})^TS^{-1}_jD_i^{cl} - \gamma ^2I
\end{align*}
Applying Schur complement, \eqref{eq:52} is equivalent to 
\begin{align}
    \begin{bmatrix}
        \Gamma_{11} &  \Gamma_{12} & 0 &0  & (K_iC_i)^T  \\
        * & \Gamma_{22} & 0 & 0 & (K_iE_i)^T \\
        * & * & \Gamma_{33} & (K_iD_i)^T& 0\\
        * & * & * & -R^{-1} & 0 \\
        * & * & * & * & -R^{-1}
    \end{bmatrix} <0 \label{eq:52_}
\end{align}
By virtue of $\sum_{j\in \mathcal{I}} p_{ij}S_j^{-1} = M_i^T \Lambda^{-1} M_i$, \eqref{eq:52_} can be expressed equivalently as
 \begin{align}
    \begin{bmatrix}
        -S^{-1}_i+Q &  0 & 0 &0  & (K_iC_i)^T  \\
        * &  -\gamma^2I  & 0 & 0 & (K_iE_i)^T \\
        * & * & -\gamma^2I  & (K_iD_i)^T& 0\\
        * & * & * & -R^{-1} & 0 \\
        * & * & * & * & -R^{-1}
    \end{bmatrix} \nonumber \\
    + Y^T\begin{bmatrix}
        \Lambda^{-1} & 0\\
        0 & \Lambda^{-1} 
    \end{bmatrix}Y <0
    \end{align}
with
    \begin{align*}
        Y = \begin{bmatrix}
        0 & 0 &  M_i D_i^{cl}&0 & 0\\
        M_iA_i^{cl} & M_i E_i^{cl} & 0 & 0& 0
    \end{bmatrix}
    \end{align*}
Applying again Schur complement yields \eqref{eq:55}. 
\begin{figure*}[!t]
   \begin{align}
        \begin{bmatrix}
        -S_i^{-1}+Q &  0 & 0 &0  & (K_iC_i)^T  & 0 & ( M_iA_i^{cl})^T  \\
        * &  -\gamma^2I  & 0 & 0 & (K_iE_i)^T &0 & ( M_iE_i^{cl})^T \\
        * & * & -\gamma^2I  & (K_iD_i)^T& 0 & (M_iD_i^{cl})^T & 0 \\
        * & * & * & -R^{-1} & 0 & 0 & 0 \\
        * & * & * & * & -R^{-1} & 0 & 0 \\
        * & * & * & * & * & -\Lambda & 0\\
        * & * & * & * & * & * & -\Lambda 
    \end{bmatrix}<0 \label{eq:55} 
\end{align}
\end{figure*}
Taking a congruent transformation, we pre- and post-multiply \eqref{eq:55} by $\text{diag}(S_i,S_i,S_i,I,I,I,I)$ and obtain the equivalence~\eqref{eq:55_}.
\begin{figure*}[!t]
   \begin{align}
        \begin{bmatrix}
        -S_i&  0 & 0 &0  & (K_iC_iS_i)^T  & 0 & ( M_iA_i^{cl}S_i)^T  & S\\
        * &  -\gamma^2S_iS_i  & 0 & 0 & (K_iE_iS_i)^T &0 & ( M_iE_i^{cl}S_i)^T & 0\\
        * & * & -\gamma^2S_iS_i  & (K_iD_i S_i)^T& 0 & (M_iD_i^{cl}S_i)^T & 0  & 0\\
        * & * & * & -R^{-1} & 0 & 0 & 0 & 0\\
        * & * & * & * & -R^{-1} & 0 & 0  & 0\\
        * & * & * & * & * & -\Lambda & 0 & 0\\
        * & * & * & * & * & * & -\Lambda  & 0\\
        * & * & * & * & * & * & *  & -Q^{-1}
    \end{bmatrix}<0 \label{eq:55_} 
\end{align}
\end{figure*}

On the other hand, using condition \eqref{eq:51}, definition \eqref{eq:8}, and $W_i = K_iY_i$, it is easy to verify that \eqref{eq:50} implies \eqref{eq:55_}. This completes the proof.\end{proof}

\begin{myremark} \label{rmk:5}
    With Theorem \ref{thm4}, the stochastic optimal guaranteed cost control (SOGCC) can be reached by solving the following optimization problem in terms of decision variables $(S_i, Y_i,W_i, \gamma, \lambda)$
    \begin{alignat}{2} \label{eq:56}
        &\text{minimize}  && \quad \gamma \nonumber \\
        &\text{subject to}  && \quad\eqref{eq:50},\eqref{eq:51} \nonumber  \\
        &  && \quad S_i>0,Y_i>0, \lambda>0, \gamma>0 
    \end{alignat}
    The resulting SOGCC gains is recovered as $K^* = W^*(Y^*)^{-1}$.
\end{myremark}
\begin{myremark} \label{rmk:6}
    It is important to note that directly solving Remark~\ref{rmk:5} can be challenging, due to the bilinear term $\gamma^2 \lambda^2$ appearing in the (2,2) and (3,3) entries of  \eqref{eq:50}. To overcome it, one practical way is to fix $\lambda$ \textit{a priori} as a sufficiently small constant, thereby transforming the constraint into a linear form.
\end{myremark}

\section{Car following control} \label{sec5}
To illustrate the usefulness of the developed results, an automated driving example of car following control (CFC) is presented. In this task,  the goal of the ego vehicle is to follow the leading vehicle with a prescribed safe distance. 

The dynamics of the ego vehicle is represented as
\begin{align}
    {x}^{ego}_1(k+1) &= x^{ego}_1(k) + hx_2^{ego}(k) \nonumber \\ 
    {x}^{ego}_2(k+1) &= x_2^{ego}(k) +  hu(k) \label{eq:58}
\end{align}
where $x_1^{ego}, x_2^{ego} \in \mathbb{R}$ denote the position and  velocity of the ego vehicle, respectively, $h\in \mathbb{R} $ the sampling interval, and $u \in \mathbb{R}$ the control input to be determined. 

The environment, in this example the leading vehicle, is modeled as
\begin{align}
    {x}^{ld}_1(k+1) &= x^{ld}_1(k) + hx_2^{ld}(k) \nonumber \\ 
    {x}^{ld}_2(k+1) &= x_2^{ld}(k) \label{eq:59}    
\end{align}
where $x_1^{ld}, x_2^{ld} \in \mathbb{R}$ denote the position and the velocity of the leading vehicle, respectively.

Let ${x}_1 = x_1^{ego} - x_1^{ld} -  \delta^{safe}$, ${x}_2 = x_2^{ego} - x_2^{ld}$, and ${x} = [{x}_1, {x}_2]^T$; note that $\delta^{safe}$ denotes the desired safety distance to be maintained. Together with \eqref{eq:58} and \eqref{eq:59}, the error dynamics of CFC and the AI-based measurement are given as
\begin{align}
    &x(k+1) = A x(k) + Bu(k)  \nonumber \\
    &y(k) =  C_i x(k) + D_i w(k) + E_i v(k) \label{eq:60}
\end{align}
with
\begin{align*}
    A = \begin{bmatrix}
        1 & h \\
        0 & 1
    \end{bmatrix}, ~ B = \begin{bmatrix}
        0 \\
        h 
    \end{bmatrix},~C_i = \begin{cases}
\text{diag}(0,1), & i =0 \\
\text{diag}(1,1), & i = 1
\end{cases} \\
D_i = \begin{cases}
\text{diag}(d_{00},d_{01}), & i = 0  \\
\text{diag}(d_{10},d_{11}), & i = 1 
\end{cases}, ~E_i = \begin{cases}
\text{diag}(e_{00},e_{01}), & i = 0   \\
\text{diag}(e_{10},e_{11}), & i = 1
\end{cases}
\end{align*}
\begin{myremark}
It is worth noting that mode~1 ($i=1$) means a successful measurement from AI-based perception, while mode~0 ($i=0$) denotes the underlying misdetection phenomenon in perception processes. This switching characteristic, i.e., $i=r(t)$, follows a given Markov chain, which is defined by the following transition probability matrix
\begin{align}
    \Xi = \begin{bmatrix}
        p_{00} & p_{01}\\
        p_{10} & p_{11}
    \end{bmatrix}
\end{align}
\end{myremark}
\begin{myremark}
The CFC system \eqref{eq:60} obtained represents a particular example of our main results. Hence, the analysis and synthesis of CFC can be systematically conducted by following the procedures outlined in Theorems~\ref{thm1}--\ref{thm4}.
\end{myremark}

\section{experiments} \label{sec6}
To verify the effectiveness of the presented control design methods, this section gives illustrative examples. The considered car following system is affected by measurement noise, bias and misdetection. The system parameters are specified as follows: $h=0.01$, and  $D_i$ are set as $d_{00} = 0.01$, $d_{01}= 0.05$, $d_{10} = 0.01$ and $d_{11}=0.05$; $E_i$ are given by $e_{00}= e_{01} = e_{10} =e_{11} = 0.01$; $w({k})$ follows the standard Gaussian distribution and $v(k) = [-1,-1]^T$. The probability transition matrix $\Xi$ are defined as $p_{00} = 0.7$, $p_{01} = 0.3$, $p_{10} = 0.2$, and $p_{11} = 0.8$. A sample path of this Markov chain is illustrated in Fig.~\ref{fig:1}, which indicates the operational status of the equipped perception system. It can be observed that a high misdetection rate occurs, and this erratic, intermittent perception can pose significant risks to an autonomous driving system. The control law \eqref{eq:7} is used, and the SSC control design is performed by virtue of Theorem \ref{thm:2}, resulting in $K_{ssc}(0)=[0,-101]$ and $K_{ssc}(1)= [-0.45, -100]$.  Setting $Q = \text{diag}(10,10)$, $R= 1$, and $\lambda=10^{-5}$ and solving problem \eqref{eq:56}, the SOGCC control gains are obtained as $K_{sogcc}(0)= [0, -3.6] $ and $K_{sogcc}(1)= [-1.22, -2.66]$.
\begin{figure}[!htbp]
\centering
\includegraphics[width=0.5\textwidth,trim=0.1cm 0.3cm 0.1cm 0.1cm, clip]{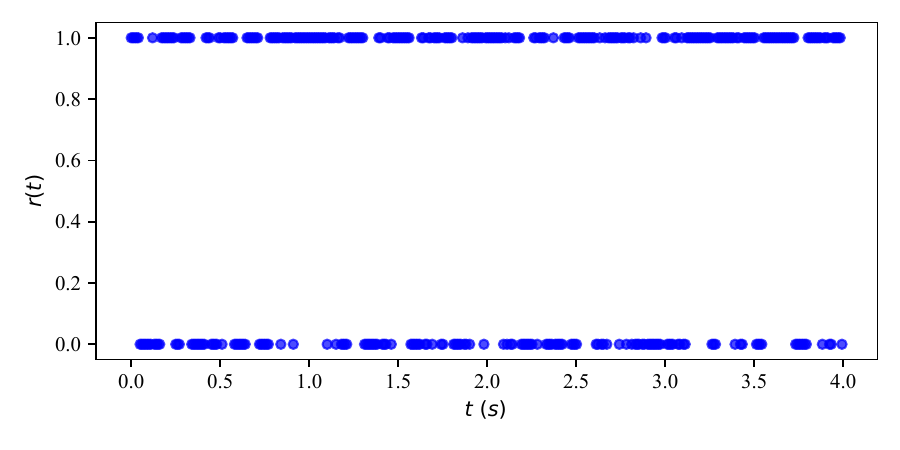}
\caption{Status of the AI-based perception system.}
\label{fig:1}
\end{figure}

The desired relative distance between the leading and ego vehicles are set to $\delta _1^d =- 5$. The initial states of the car following system are given as $[x_1^{ego}(0),x_2^{ego}(0)]^T = [0,1]^T$, $[x_1^{ld}(0),x_2^{ld}(0)]^T = [10,5]^T$. The intelligent driving model (IDM), a widely acknowledged method in the AV field, is used for comparison. Given the stochastic nature of the scenario, the simulations are conducted over 200 independent trials. We present both mean values and standard deviation for each quantity of interest with the shaded plots to illustrate the variation across trials. The results are reported in Figs.~\ref{fig:2}--\ref{fig:4}. As shown in Fig.~\ref{fig:2},  only the root mean square error (RMSE) of the SOGCC approach converges within the simulation period and reaches a very small steady-state error. This demonstrates the superior performance of the SOGCC method. While the SSC approach may also exhibit a convergence in RMSE, much more time is required. In contrast, the IDM method fails to converge under the high misdetection conditions. The non-convergence may lead to highly risky and unsafe driving, such as collisions. This safety concern is further evidenced in Fig.~\ref{fig:3}, where the shaded red region indicates the collision zone. It is clear that the trajectories generated by the IDM policy intersect this region, showing multiple potential collisions. In addition, the behavior of IDM also becomes extremely unpredictable, as indicated by its large standard deviation.   On the other hand, the SOGCC approach converges smoothly to the desired safe distance with high confidence despite the adverse misdetection and noisy, biased perception, whereas the SSC converges much slower, driving more conservatively yet without any collisions. Fig.~\ref{fig:4} illustrates the control actions produced by the three methods. The SSC policy exhibits highly fluctuating and excessive actions, which will definitely lead to discomfort for passengers. In comparison, the SOGCC method generates the smoothest and most reasonable control signals, achieving a favorable trade-off between performance, safety, and comfort.
\begin{figure}[!htbp]
\centering
\includegraphics[width=0.5\textwidth,trim=0.2cm 0.2cm 0cm 0.2cm, clip]{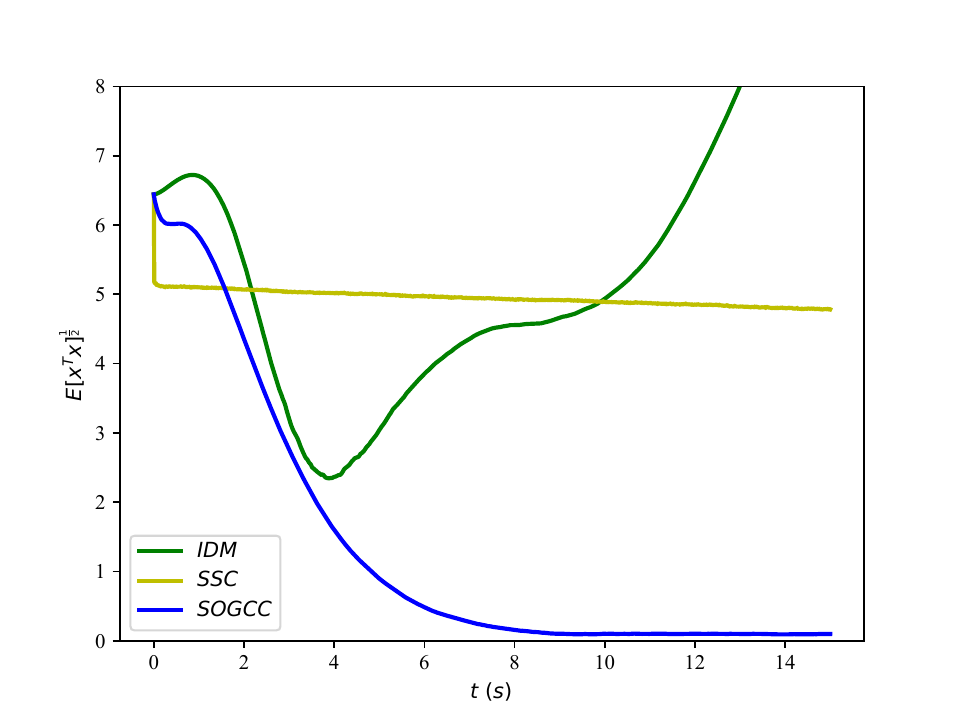}
\caption{Root mean square error performance of three different approaches.}
\label{fig:2}
\end{figure}

\begin{figure}[!htbp]
\centering
\includegraphics[width=0.47\textwidth,trim=0.2cm 0.2cm 0cm 0.2cm, clip]{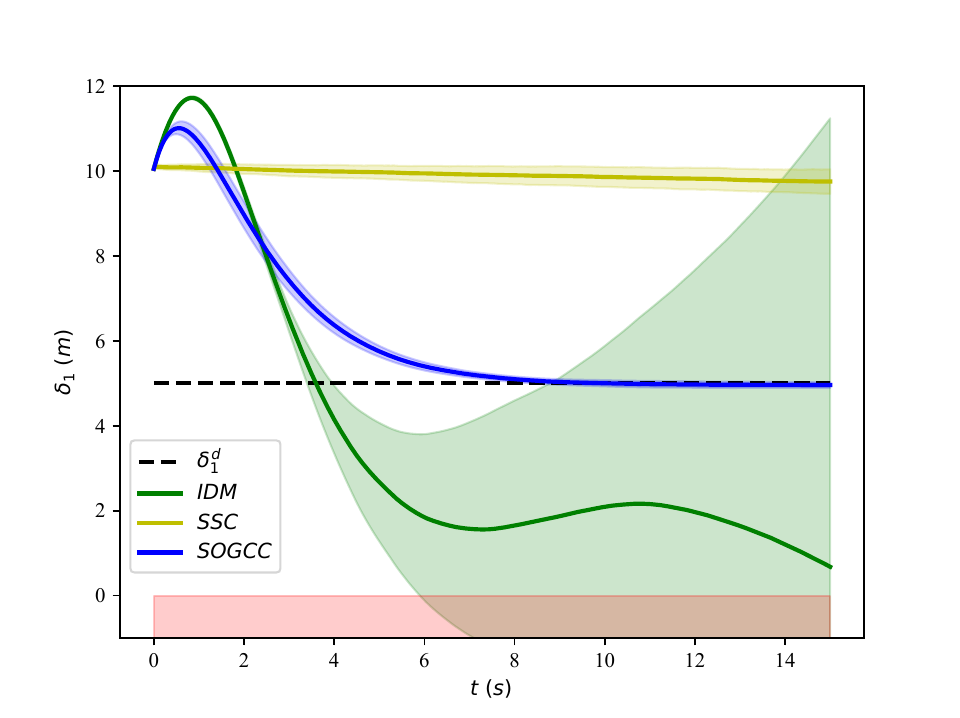}
\caption{ Relative distance evolution between leading and ego vehicles.}
\label{fig:3}
\end{figure}

\begin{figure}[!htbp]
\centering
\includegraphics[width=0.47\textwidth]{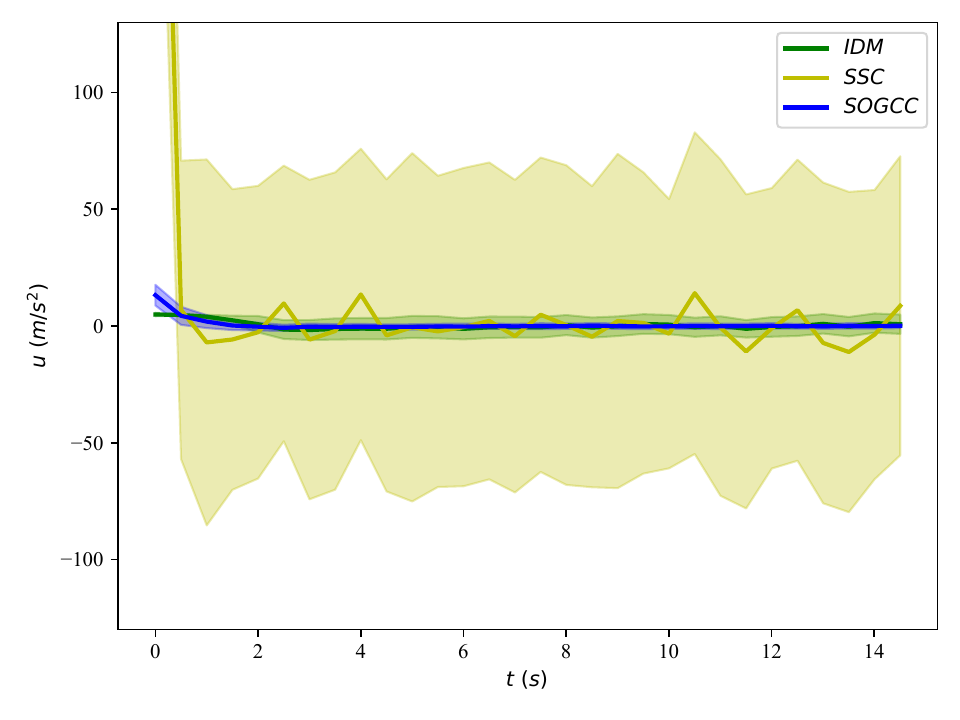}
\caption{The control actions generated by three control schemes.}
\label{fig:4}
\end{figure}

Based on the experiment results, it is evident that adverse perception conditions, particularly when a high rate of misdetection presents, can significantly degrade the performance of ADSs. Failure to address it can introduce high risks and lead to unsafe driving behavior, as exemplified  in Fig.~\ref{fig:3}. The proposed two control methods effectively handle this challenge by  incorporating misdetection explicitly into the control design, thus greatly improving the reliability of autonomous driving. In particular, it turns out that the proposed stochastic optimal guaranteed cost control can simultaneously maintain robustness and performance even in the presence of various perception uncertainties.  

\section{conclusion} \label{sec7}
This paper models AI-based ADSs as a new class of control systems that are affected by unreliable perception, including stochastic jumping, noise, and unknown, time-varying bias. Sufficient conditions are developed for both stochastic stability analysis and stabilizing control synthesis.  A novel concept of stochastic guaranteed costs is introduced to quantify performance and robustness of systems subject to heterogeneous sources of perception uncertainty. Furthermore, an efficient convex approximation is proposed to solve stochastic optimal guaranteed cost control which is not easy to resolve in existing frameworks. The techniques are validated through a car following scenario. The  experimental results show the effectiveness of the proposed SOGCC method in balancing reliability, performance, and passenger comfort under adverse sensing and perception conditions.



\bibliographystyle{IEEEtran}
\bibliography{reference1}

\end{document}